\documentclass{revtex4}
\usepackage{amsmath,amsthm}
\usepackage{rotating}
\usepackage{multirow}
\usepackage{graphicx}

\usepackage{amsfonts}
\newtheorem{theorem}{Theorem}[section]
\newtheorem{lemma}[theorem]{Lemma}
\newtheorem{proposition}[theorem]{Proposition}

\theoremstyle{remark}
\newtheorem{remark}[theorem]{Remark}

\theoremstyle{definition}
\newtheorem{definition}[theorem]{Definition}

\theoremstyle{example}
\newtheorem{example}[theorem]{Example}

\theoremstyle{notation}

\newcommand{\bra}[1]{\langle#1|}
\newcommand{\ket}[1]{|#1\rangle}

\begin{document}

\title{Renormalization of total sets of states into generalized bases with a resolution of the identity}            
\author{A. Vourdas}
\affiliation{Department of Computer Science,\\
University of Bradford, \\
Bradford BD7 1DP, United Kingdom\\a.vourdas@bradford.ac.uk}

\begin{abstract}
A total set of states for which we have no resolution of the identity (a `pre-basis'), is considered in a finite dimensional Hilbert space.
A dressing formalism renormalizes them into density matrices which resolve the identity, and makes them 
a `generalized basis', which is practically useful. 
The dresssing mechanism is inspired by Shapley's methodology in cooperative game theory, and it uses M\"obius transforms.
There is non-independence and redundancy in these generalized bases,
which is quantified with a Shannon type of entropy.
Due to this redundancy, calculations based on generalized bases, are sensitive to physical changes and robust in the presence of noise.
For example, the representation of an arbitrary vector in such generalized bases, is robust when noise is inserted in the coefficients.
Also in a physical system with ground state which changes abruptly at some value of the coupling constant,
the proposed methodology detects such changes, even when noise is added to the parameters in the Hamiltonian of the system.
\end{abstract}
\maketitle
\section{Introduction}
Redundancy is important for error correction.
Without redundancy in our language (quantified by Shannon \cite{S} and later by many others)
we would not be able to communicate because a minor spelling mistake would change completely the meaning.
The analogue of this in the context of Hilbert spaces is that calculations based on orthonormal bases are sensitive to noise.
In contrast, calculations based on total (or overcomplete) sets of vectors that can be used as generalized bases, are 
much less sensitive to noise.
A set $\Sigma$ of vectors is called total, if there is no vector in the Hilbert space which is orthogonal to all vectors in $\Sigma$.

A total set of vectors can be used as a generalized basis, only if there is a resolution of the identity in terms 
of them, which can be used to expand an arbitrary vector in terms of the vectors in the total set.
In the present paper we consider a $d$-dimensional Hilbert space $H_d$, and an arbitrary total set of $n> d$ vectors (for which 
in general we have no resolution of the identity).
We renormalize them into a set of $n$ mixed states (density marices), that resolve the identity.

The renormalization formalism is analogous to the Shapley methodology in cooperative game theory\cite{G1,G2,G3,G4}.
In a recent paper \cite{V2017} we have used this methodology mainly with the set of $n=d^2$ coherent states (which is a special case of a total set), and we only discussed briefly the application of the formalism to an arbitrary total set.
In the present paper we expand the use of the formalism with an arbitrary total set, as follows:
\begin{itemize}
\item
The formalism is presented directly in a quantum context.
The analogy between the  Shapley methodology in cooperative game theory and our approach in a quantum context, has been discussed in detail in 
\cite{V2017} and is not discussed here.
We note that cooperative game theory uses scalar quantities, while quantum mechanics uses matrices. 

\item
The formalism leads to $n> d$ density matrices $\sigma (i)$, which resolve the identity and which can be used as a generalized basis.
The term `generalized basis', reflects:
\begin{itemize}
\item
the fact that it consists of density matrices (i.e., vectors with probabilities attached to them)
\item
their non-independence (the number of them is greater than the dimension of the space).
The non-independence and redundancy in this generalized basis, is quantified with a Shannon type of entropy which is shown to take values in the 
interval $(\log n-\log d, \log n)$. The merit of having this redundancy, is that it
makes calculations with generalized bases, sensitive to physical changes and immune to noise.
\end{itemize}
\item
Coherent states are uniformly distributed in phase space, and the corresponding renormalized density matrices $\sigma (i)$ (studied in  \cite{V2017}) have 
strong properties related to coherence, e.g., they are related to each other through displacement transformations.
However from a practical point of view, for large dimension $d$ 
the calculation is tedious (in a $d$-dimensional space, there are $d^2$ coherent states, which lead to $d^2$ renormalized density matrices $\sigma (i)$).
The formalism discussed in this paper is general, and we can take $n$ slightly larger than $d$, so that we have the merits of redundancy with fewer 
renormalized density matrices $\sigma (i)$, and a simpler calculation.
Of course, in this general case the renormalized density matrices $\sigma (i)$, have weaker properties than in the case of coherent states.

\item
The emphasis in this paper is in the applications of the formalism, as follows:
\begin{itemize}
\item
It is shown that the representation of a vector in our generalized bases is robust in the presence of noise, in the sense that 
addition of random numbers in the coefficients does not change the vector significantly.
\item
The formalism is applied to the study of the ground state 
(i.e., the eigenstate corresponding to the lowest eigenvalue),
of physical systems. We consider a system in which the ground state changes 
abruptly at some value of the coupling constant.
We show that our generalized bases can detect such changes even in the presence of noise.
In large (ideally infinite) systems, such an abrupt change of the ground state is associated with a phase transition.
\end{itemize}
\end{itemize}

The whole area of coherent states, POVMs (positive operator valued measures) and 
frames and wavelets (e.g., \cite{C1,C2,C3}), are a kind of generalized bases, 
and calculations that use them are robust in the presence of noise, due to redundancy.
An arbitrary state can be expanded in terms of coherent states or POVMs, because of a resolution of the identity.
In frames we have no exact resolution of the identity, but we have lower and upper bounds to it.
In our case we start from a total set of $n> d$ states, which we renormalize and we get $n$ density matrices $\{\sigma (i)\}$, 
that resolve the identity.
They can be used as a generalized basis, which is robust in the presence of noise.

In section 2 we define various quantities and explain the notation. In section 3, we present briefly the M\"obius transforms.
We discussed them in a different context in \cite{V1,V2} and here  we only give briefly the relevant formulas, together with a new proposition 
on the trace of these operators (proposition \ref{b135}),  which is used later.

In section 4, we show how to renormalize a total set of vectors into a generalized basis, which resolves the identity.
The starting point is a resolution of the identity that contains projectors associated with the vectors in the total set, and M\"obius operators.
Using an approach inspired by the  Shapley methodology in cooperative game theory, we assign the M\"obius operators to the projectors, and convert them into density matrices that resolve the identity.

In section 5, the redundancy in the generalized bases, is quantified with a Shannon type of entropy.
In section 6, we use the generalized bases, to represent a vector in $H_d$, with $n$ components. 
We then add noise to these components, and reconstruct the original vector.
It is shown that the error in this reconstruction, is smaller in the case of our generalized bases, than in the case of orthonormal bases.

In section 7, we consider a physical system with two-dimensional Hilbert space and with Hamiltonian $\theta (\lambda)$, 
 whose ground state changes abruptly at some value of the coupling constant $\lambda$. 
 Such a system is often used as an approximation to an infinite-dimensional system, 
which due to low energy, operates in the subspace of the lowest two states.
Many of the experimentally available qubits are of this type (e.g., the superconducting qubits).
 We define the concept of location index ${\cal L}[\theta (\lambda)]$ of $\theta (\lambda) $, 
with respect to a generalized basis $\{\sigma (i)\}$. 
We then define comonotonicity intervals of the coupling parameter $\lambda$, within which the location index ${\cal L}[\theta (\lambda)]$
remains constant, and associate them with mild changes in the physical system. 
Crossing points from one comonotonicity interval to another, indicate a possible drastic change in the ground state of the system.
We show that the method works well, even when we add noise in the parameters of the Hamiltonian.

 We conclude in section 8, with a discussion of our results.
   
\section{Preliminaries}
We consider a $d$-dimensional Hilbert space $H_d$, and an orthonormal basis of `position states' which we denote as
$\ket {X;\alpha}$. Here $a \in {\mathbb Z}(d)$ (the integers modulo $d$), and the $X$ in the notation is not a variable but it simply indicates position states.  
The Fourier transform is defined as
\begin{eqnarray}
F=\frac{1}{\sqrt{d}}\sum \omega (\alpha \beta)\ket {X;\alpha}\bra {X;\beta};\;\;\;\omega(\alpha)=\exp \left(i\frac{2\pi \alpha}{d}\right).
\end{eqnarray}
\begin{definition}
A `pre-basis' in the $d$-dimensional Hilbert space $H_d$, is a set of $n> d$ states 
\begin{eqnarray}\label{sig}
\Sigma =\{\ket{i}\;|\;i\in \Omega\};\;\;\;\Omega=\{1,...,n\}
\end{eqnarray}
such that:
\begin{itemize}
\item
Any subset of $d$ of these states, are linearly independent. 
\item
$\Sigma$ and also any of its subsets with $r\ge d$ of these states, are total sets.
\item
In general, we have no resolution of the identity in terms of these $n$ states.
\end{itemize}
\end{definition}
We call the
\begin{eqnarray}\label{red}
{\cal R}=\frac{n-d}{d}> 0,
\end{eqnarray}
redundancy index.
For coherent states ${\cal R}=d-1$, and for large $d$ this is a large redundancy.
The formalism in this paper is general, but from a practical point of view it should be used with 
positive but small values of  ${\cal R}$.

Let $H(A)=H(i_1,...,i_r)$ be the subspace of $H_d$ spanned by the states $\ket{i_1},...\ket{i_r}$:
\begin{eqnarray}
H(A)=H(i_1,...,i_r)={\rm span}\{\ket{i_1},...\ket{i_r}\};\;\;\;A=\{i_1,...,i_r\}\subseteq \Omega.
\end{eqnarray}
If $r<d$ then $H(A)$ is an $r$-dimensional subspace of $H_d$.
If $r\ge d$, then $H(A)=H_d$.
We call $\Pi(A)=\Pi(i_1,...,i_r)$ the projector to the subspace $H(A)$.
In general
\begin{eqnarray}\label{3xc}
\Pi(i_1,...,i_r)\ne \Pi(i_1)+...+\Pi(i_r).
\end{eqnarray}
Only if the kets $\ket{i_1},...\ket{i_r}$ are orthogonal to each other, we get equality in this equation.
Also, in general there is no constant $\mu$ such that 
\begin{eqnarray}\label{coh}
\mu [\Pi(i_1)+...+\Pi(i_r)]\ne {\bf 1}.
\end{eqnarray}
In special cases (e.g., with the total set of $n=d^2$ coherent states), we might get equality in Eq.(\ref{coh}).

Cooperative game theory renormalizes the individual contribution of a player, by adding his contribution to aggregations of players.
Similarly, we renormalize $\Pi(i)=\ket{i}\bra{i}$ by adding to it the contributions of the state $i$, to aggregations of states described with projectors $\Pi(A)$ where $i\in A$. 

\section{M\"obius transforms}\label{IVA}
M\"obius transforms have been introduced by Rota\cite{R,R1}. They are a generalization of the 
`inclusion-exclusion\rq{} principle in set theory, which gives the cardinality of the union of overlapping sets.
M\"obius transforms find the overlaps between sets, and thus avoid the double-counting.
Rota generalized them to partially ordered structures, and here we use them with projectors to Hilbert spaces. 

In refs\cite{V1,V2} we have discussed M\"obius transforms in a different context,
and in this section we only give briefly the relevant formulas. 
The M\"obius transform of the coherent projectors $\Pi(A)$, is given by:
\begin{eqnarray}\label{m11}
{\mathfrak D} (B)=\sum _{A\subseteq B} (-1)^{|A|-|B|}\Pi(A);\;\;\;\;A,B\subseteq\Omega.
\end{eqnarray}
The inverse M\"obius transform is
\begin{eqnarray}\label{m13}
\Pi (A)=\sum _{B\subseteq A}{\mathfrak D} (B).
\end{eqnarray}
Some examples are:
\begin{eqnarray}\label{m12}
&&{\mathfrak D} (1)=\Pi(1);\;\;\;{\mathfrak D} (1,2)=\Pi(1,2)-\Pi(1)-\Pi(2)\nonumber\\
&&{\mathfrak D} (1,2,3)=\Pi(1,2,3)-\Pi(1,2)-\Pi(1,3)-\Pi(2,3)+\Pi(1)+\Pi(2)+\Pi(3),
\end{eqnarray}
and then
\begin{eqnarray}
&&\Pi(1,2)={\mathfrak D} (1,2)+{\mathfrak D} (1)+{\mathfrak D} (2)\nonumber\\
&&\Pi(1,2,3)={\mathfrak D} (1,2,3)+{\mathfrak D} (1,2)+{\mathfrak D} (1,3)+{\mathfrak D} (2,3)
+{\mathfrak D} (1)+{\mathfrak D} (2)+{\mathfrak D} (3).
\end{eqnarray}
We note that if the total set $\Sigma$ consists of an orthonormal set of $d$ states,
then all the ${\mathfrak D} (B)$ with $|B|\ge 2$, are zero.

\begin{proposition}\label{b135}
The trace of ${\mathfrak D} (B)$ is given by
\begin{eqnarray}
&&{\rm Tr}[{\mathfrak D} (B)]=1;\;\;{\rm if}\;\;|B|=1\nonumber\\
&&{\rm Tr}[{\mathfrak D} (B)]=0;\;\;{\rm if}\;\;2\le |B|\le d\nonumber\\
&&{\rm Tr}[{\mathfrak D} (B)]=(-1)^{d-|B|}\begin{pmatrix}
|B|-2\\
d-1\\
\end{pmatrix};\;\;{\rm if}\;\;|B|\ge d+1.
\end{eqnarray}
\end{proposition}
\begin{proof}
We first point out that
\begin{eqnarray}
&&{\rm Tr}[\Pi(A)]=|A|;\;\;{\rm if}\;\;|A|\le d\nonumber\\
&&{\rm Tr}[\Pi(A)]=d;\;\;{\rm if}\;\;|A|> d
\end{eqnarray}
In the sum of Eq.(\ref{m11}), there are 
$\begin{pmatrix}
|B|\\
k\\
\end{pmatrix}$
sets $A$ with the same cardinality $|A|=k$. Therefore if $2\le |B|\le d$ we get \cite{GR}
\begin{eqnarray}
{\rm Tr}[{\mathfrak D} (B)]&=&(-1)^{-|B|}\sum _{k=1}^{|B|}(-1)^{k}\begin{pmatrix}
|B|\\
k\\
\end{pmatrix}k
=(-1)^{-|B|+1}|B|\sum _{k=1}^{|B|}(-1)^{k-1}\begin{pmatrix}
|B|-1\\
k-1\\
\end{pmatrix}\nonumber\\&=&
(-1)^{-|B|+1}|B|\sum _{k=0}^{|B|-1}(-1)^{k}\begin{pmatrix}
|B|-1\\
k\\
\end{pmatrix}=
(1-1)^{|B|-1}=
0.
\end{eqnarray}
In the case $|B|\ge d+1$ we get
\begin{eqnarray}
{\rm Tr}[{\mathfrak D} (B)]=(-1)^{-|B|}\sum _{k=1}^{d}(-1)^{k}\begin{pmatrix}
|B|\\
k\\
\end{pmatrix}k+d
(-1)^{-|B|}\sum _{k=d+1}^{|B|-1}(-1)^{k}\begin{pmatrix}
|B|\\
k\\
\end{pmatrix}
+d
\end{eqnarray}
But
\begin{eqnarray}
(-1)^{-|B|}\sum _{k=1}^{d}(-1)^{k}\begin{pmatrix}
|B|\\
k\\
\end{pmatrix}k=
(-1)^{-|B|}|B|\sum _{k=1}^{d}(-1)^{k}\begin{pmatrix}
|B|-1\\
k-1\\
\end{pmatrix}=
(-1)^{d-|B|}|B|\begin{pmatrix}
|B|-2\\
d-1\\
\end{pmatrix}
\end{eqnarray}
Also (use formula 0.151.4 in \cite{GR})
\begin{eqnarray}
\sum _{k=0}^{|B|-1}(-1)^{k}\begin{pmatrix}
|B|\\
k\\
\end{pmatrix}=(-1)^{|B|-1};\;\;\;\;
\sum _{k=0}^{d}(-1)^{n}\begin{pmatrix}
|B|\\
k\\
\end{pmatrix}=
(-1)^{d}\begin{pmatrix}
|B|-1\\
d\\
\end{pmatrix}
\end{eqnarray}
Combining these results we prove that
\begin{eqnarray}
d(-1)^{-|B|}\sum _{k=d+1}^{|B|-1}(-1)^{k}\begin{pmatrix}
|B|\\
k\\
\end{pmatrix}=-d-
d(-1)^{d-|B|}\begin{pmatrix}
|B|-1\\
d\\
\end{pmatrix}=
-d-(-1)^{d-|B|}(|B|-1)\begin{pmatrix}
|B|-2\\
d-1\\
\end{pmatrix}
\end{eqnarray}
and then prove the last relation in the proposition.
\end{proof}

M\"obius transforms are intimately related to commutators that involve the projectors, e.g.\cite{V1,V2},
\begin{eqnarray}\label{375}
&&[\Pi(i),\Pi(j)]={\mathfrak D} (i,j)[\Pi(i)-\Pi(j)]\nonumber\\
&&[[\Pi(i), \Pi(k)], \Pi(j)]=\Pi(j){\mathfrak D} (i,j,k)[\Pi(i)-\Pi(k)]+[\Pi(i)-\Pi(k)]{\mathfrak D} (i,j,k)\Pi(j).
\end{eqnarray}
Working with M\"obius operators is equivalent to taking into account the non-commutativity of the projectors $\Pi (i)$.

\section{Renormalization of a pre-basis into a generalized basis}\label{gen}

\begin{definition}
A generalized basis in $H_d$ is a set of $n> d$ density matrices $\{\sigma (i)\}$ which obey the relation
\begin{eqnarray}
\sum \sigma (i)=\lambda {\bf 1},
\end{eqnarray}
where $\lambda$, is a constant.
\end{definition}

In this section we renormalize an arbitrary pre-basis into a generalized basis, using M\"obius transformations.
If $A$ in Eq.(\ref{m13}) is the total set $\Omega$ of Eq.(\ref{sig}), then $\Pi(\Omega)={\bf 1}$ and we get
\begin{eqnarray}\label{m134}
\sum _{B\subseteq \Omega}{\mathfrak D} (B)=\sum _{i\in \Omega}\Pi(i)+\sum _{i,j}{\mathfrak D} (i,j)+\sum _{i,j,k}{\mathfrak D} (i,j,k)+...={\bf 1}.
\end{eqnarray}
This is a resolution of the identity that involves not only the projectors $\Pi(i)$, but also the M\"obius operators
${\mathfrak D} (i,j)$, ${\mathfrak D} (i,j,k)$, etc.
The ${\mathfrak D} (i,j)=\Pi(i,j)-\Pi(i)-\Pi(j)$ `belongs' to both states $i,j$, and Eq.(\ref{375}) shows that it is related to the 
commutator $[\Pi(i),\Pi(j)]$.
We divide this `joint property' equally to all its `owners\rq{}: half of it to $i$ and the other half to $j$.
Similarly, the  ${\mathfrak D} (i,j,k)$ `belongs' to the states labeled with $i,j,k$, and we allocate a third of it to each of these three states; etc.
So we resolve the identity in Eq.(\ref{m134}) as
\begin{eqnarray}\label{m135}
\sum _{i\in \Omega}\tau (i)={\bf 1};\;\;\;
\tau (i)=\sum  _{B\subseteq \Omega \setminus \{i\}}\frac{{\mathfrak D} (B\cup \{i\})}{|B\cup \{i\}|}=\Pi(i)+\frac{1}{2}\sum _{j}{\mathfrak D} (i,j)+\frac{1}{3}\sum _{j,k}{\mathfrak D} (i,j,k)+...
\end{eqnarray}
In $\tau(i)$ the summations are over all aggregations that involve the state $i$.
We will show that the $\tau(i)$ with appropriate normalization are density matrices.

The following lemma expresses $\tau(i)$ as a sum of projectors, and will be used below to prove that 
the $\tau (i)$ are positive semidefinite operators.
It has been proved indirectly in ref\cite{V2017}, through analogy with similar results in cooperative game theory.
Below we give a direct combinatorial proof.
\begin{lemma}
Let $\varpi(i|A)$ be the projectors
\begin{eqnarray}
\varpi(i|A)=\Pi (\{i\}\cup A)-\Pi(A);\;\;\;A\subseteq \Omega \setminus  \{i\}.
\end{eqnarray}
The $\tau (i)$ can be expressed as
\begin{eqnarray}\label{m130}
\tau (i)=\frac{1}{n}\sum _{A\subseteq \Omega \setminus \{i\}}\begin{pmatrix}n-1\\|A|\\\end{pmatrix}^{-1}\varpi (i|A)
\end{eqnarray}
\end{lemma}
\begin{proof}
We count the number of projectors $\Pi(A)$ with $A\subseteq \Omega \setminus  \{i\}$, in the 
right hand side of Eq.(\ref{m135}).
There are $\begin{pmatrix}n-1-|A|\\k\\\end{pmatrix}$ Mobius operators ${\mathfrak D}(B\cup \{i\})$, with  $A\subseteq B \subseteq \Omega \setminus \{i\}$, and
$|B|=|A|+k$. Each of them contains $\Pi(A)$ with sign $(-1)^{k+1}$, and also $\Pi(A\cup \{i\})$ with sign $(-1)^{k}$. Therefore 
the number of projectors $\Pi(A)$  in the right hand side of Eq.(\ref{m135}), is 
\begin{eqnarray}\label{m136}
-\sum _{k=0}^{n-1-|A|}(-1)^{k}\begin{pmatrix}n-1-|A|\\k\\\end{pmatrix}\frac{1}{|A|+k+1}=-\frac{1}{n}\begin{pmatrix}n-1\\|A|\\\end{pmatrix}^{-1}
\end{eqnarray}
We used here the combinatorial relation
\begin{eqnarray}
\sum _{i=1}^{N}(-1)^{i-1}\begin{pmatrix}N-1\\i-1\\\end{pmatrix}\frac{1}{w+i}=\frac{w!(N-1)!}{(w+N)!}
\end{eqnarray}
The number of projectors $\Pi(A\cup \{i\})$ is also given by Eq.(\ref{m136}), but with a plus sign.
This proves Eq.(\ref{m130}).
\end{proof}
\begin{remark}
For a given $A$, the projectors $\varpi(i|A)$, $\Pi (\{i\}\cup A)$, $\Pi(A)$ commute with each other.
Measurement with $\varpi(i|A)$ will give the result `yes', if the measurement $\Pi (\{i\}\cup A)$ gives `yes', and the 
measurement $\Pi(A)$ gives `no'.
Measurement with $\varpi(i|A)=\Pi (\{i\}\cup A)-\Pi(A)$ gives the probability that the state of the system belongs to the space $H(\{i\}\cup A)$ but does not belong to the space $H(A)$.
\end{remark}
\begin{proposition}
\begin{itemize}
\item[(1)]
The $\tau (i)$ are  positive-semidefinite Hermitian matrices.
\item[(2)]
The $\sigma (i)$ given by
\begin{eqnarray}\label{5gx}
\sigma (i)=\frac{n}{d}\tau (i);\;\;\; \frac{d}{n}\sum _{i\in \Omega} \sigma(i)={\bf 1},
\end{eqnarray}
are density matrices which resolve the identity.
\item[(3)]
If the total set $\Sigma$ consists of an orthonormal set of $d$ states, then $\sigma (i)=\Pi(i)$.
\end{itemize}
\end{proposition}
\begin{proof}
\begin{itemize}
\item[(1)]
$\tau (i)$ is given in Eq.(\ref{m130}), as a sum of projectors with positive coefficients, and this proves that they are positive semidefinite 
Hermitian matrices.
\item[(2)]
There are $\begin{pmatrix}n-1\\|A|\\\end{pmatrix}$ projectors  $\varpi(i|A)$ with the same cardinality $|A|$ of $A$.
Therefore
\begin{eqnarray}
{\rm Tr}[\tau (i)]=\frac{1}{n}\sum _{A\subseteq \Omega \setminus \{i\}}\begin{pmatrix}n-1\\|A|\\\end{pmatrix}^{-1}\begin{pmatrix}n-1\\|A|\\\end{pmatrix}=
\frac{1}{n}\sum _{|A|=0}^{d-1}1=\frac{d}{n}.
\end{eqnarray}
An alternative proof will be to use Eq.(\ref{m135}) and proposition \ref{b135}.
It is seen that the trace of $\tau (i)$ does not depend on $i$.
It follows that the $\sigma (i)=\frac{n}{d}\tau (i)$, are density matrices.
\item[(3)]
We have explained earlier, that if  the total set $\Sigma$ consists of an orthonormal set of $d$ states, then
the Mobius operators ${\mathfrak D}(B)=0$ for $|B|\ge 2$. Therefore in this case $\sigma (i)=\Pi(i)$. 
\end{itemize}
\end{proof}
\begin{proposition}
Let $\{\Pi(i)\}$ be  a pre-basis , $\{\sigma (i)\}$ the corresponding generalized basis, and $U$ a unitary transformation. 
Then, the generalized basis corresponding to the pre-basis 
$\{\Pi _U(i)=U\Pi(i)U^{\dagger}\}$ is $\{\sigma _U(i)=U\sigma (i)U^{\dagger}\}$, and obeys the resolution of the identity
\begin{eqnarray}\label{5vw}
\frac{d}{n}\sum _{i\in \Omega} \sigma_U(i)={\bf 1}.
\end{eqnarray}
In particular, if $F$ is the Fourier transform, the generalized basis corresponding to the pre-basis 
$\{{\widetilde \Pi} (i)=F\Pi(i)F^{\dagger}\}$ is $\{{\widetilde \sigma} (i)=F\sigma (i)F^{\dagger}\}$, and obeys the resolution of the identity
\begin{eqnarray}
\frac{d}{n}\sum _{i\in \Omega} {\widetilde \sigma}(i)={\bf 1}.
\end{eqnarray}
\end{proposition}
\begin{proof}
Eq.(\ref{m11}) shows that if ${\mathfrak D}(B)$ are the M\"obius transforms of the projectors $\Pi(A)$,  
then the $U{\mathfrak D}(B)U^{\dagger}$ are the M\"obius transforms of the projectors $U\Pi(A)U^{\dagger}$.
Then from Eq.(\ref{m135}), follows the statement in the proposition.
Acting with $U$ and $U^\dagger$ on both sides of Eq.(\ref{5gx}), we prove the resolution of the identity in Eq.(\ref{5vw}).

The Fourier transform is a special case of a unitary transformation.
\end{proof}

\subsection{Example I}\label{sec1}
In $H_2$ we consider the total set of states:
\begin{eqnarray}\label{set1}
\Sigma=\left \{\ket{X;0}, \frac{1}{\sqrt 5}(\ket{X;0}+2i\ket{X;1}),  \frac{1}{\sqrt 2}(\ket{X;0}+\ket{X;1})\right\}.
\end{eqnarray}
In this case $n=3$, and
\begin{eqnarray}\label{27}
&&{\mathfrak D}(1,2)=\frac{1}{5}
\begin{pmatrix}
-1&2i\\
-2i&1\\
\end{pmatrix};\;\;\;
{\mathfrak D}(1,3)=\frac{1}{2}
\begin{pmatrix}
-1&-1\\
-1&1\\
\end{pmatrix};\;\;\;
{\mathfrak D}(2,3)=\frac{1}{10}
\begin{pmatrix}
3&-5+4i\\
-5-4i&-3\\
\end{pmatrix}\nonumber\\&&
{\mathfrak D}(1,2,3)=\frac{1}{10}
\begin{pmatrix}
-3&5-4i\\
5+4i&-7\\
\end{pmatrix}
\end{eqnarray}
Then
\begin{eqnarray}\label{eq30}
\frac{2}{3}\sigma  (1)&=&\Pi (1)+\frac{1}{2}[{\mathfrak D} (1,2)+{\mathfrak D} (1,3)]+\frac{1}{3}{\mathfrak D} (1,2,3)
\end{eqnarray}
and similarly for $\sigma (2), \sigma (3)$.
Therefore
\begin{eqnarray}\label{sigma}
&&\Pi(1)=\begin{pmatrix}
1&0\\
0&0\\
\end{pmatrix}\rightarrow \sigma  (1)=
\begin{pmatrix}
0.825&-0.125+0.100i\\
-0.125-0.100i&0.175\\
\end{pmatrix};\nonumber\\
&&\Pi(2)=\frac{1}{5}\begin{pmatrix}
1&-2i\\
2i&4\\
\end{pmatrix}\rightarrow \sigma  (2)=
\begin{pmatrix}
0.225&-0.125-0.200i\\
-0.125+0.200i&0.775\\
\end{pmatrix}\nonumber\\
&&\Pi(3)=\frac{1}{2}\begin{pmatrix}
1&1\\
1&1\\
\end{pmatrix}\rightarrow \sigma  (3)=
\begin{pmatrix}
0.450&0.250+0.100i\\
0.250-0.100i&0.550\\
\end{pmatrix}
\end{eqnarray}
The resolution of the identity is
\begin{eqnarray}\label{aw4}
\frac{2}{3}[\sigma  (1)+\sigma  (2)+\sigma  (3)]={\bf 1}.
\end{eqnarray}
We also give the Fourier transform of this generalized basis:
\begin{eqnarray}
&&{\widetilde \Pi}(1)=\frac{1}{2}\begin{pmatrix}
1&1\\
1&1\\
\end{pmatrix}\rightarrow {\widetilde \sigma}  (1)=
\begin{pmatrix}
0.375&0.325-0.100i\\
0.325+0.100i&0.625\\
\end{pmatrix};\nonumber\\
&&{\widetilde \Pi}(2)=\begin{pmatrix}
0.5&-0.3+0.4i\\
-0.3-0.4i&0.5\\
\end{pmatrix}\rightarrow {\widetilde \sigma}  (2)=
\begin{pmatrix}
0.375&-0.275+0.200i\\
-0.275-0.200i&0.625\\
\end{pmatrix}\nonumber\\
&&{\widetilde \Pi}(3)=\frac{1}{2}\begin{pmatrix}
1&0\\
0&0\\
\end{pmatrix}\rightarrow {\widetilde \sigma}  (3)=
\begin{pmatrix}
0.750&-0.050-0.100i\\
-0.050+0.100i&0.250\\
\end{pmatrix}.
\end{eqnarray}
The resolution of the identity in this case is
\begin{eqnarray}\label{aw4}
\frac{2}{3}[{\widetilde \sigma}  (1)+{\widetilde \sigma}  (2)+{\widetilde \sigma}  (3)]={\bf 1}.
\end{eqnarray}
\subsection{Example II}\label{sec2}
In $H_2$ we consider the total set of states:
\begin{eqnarray}\label{set2}
\Sigma=\left \{\ket{X;0}, \frac{1}{\sqrt 5}(\ket{X;0}+2i\ket{X;1}),  \frac{1}{\sqrt 2}(\ket{X;0}+\ket{X;1}), \frac{1}{\sqrt 5}(\ket{X;0}+2\ket{X;1})\right\}.
\end{eqnarray}
In comparison to the previous example, we added here the fourth vector.
In this case $n=4$. The ${\mathfrak D}(1,2)$, ${\mathfrak D}(1,3)$, ${\mathfrak D}(2,3)$, and ${\mathfrak D}(1,2,3)$, are the same as in Eq.(\ref{27}).
In addition to them, we have here the
\begin{eqnarray}
&&
{\mathfrak D}(1,4)=\frac{1}{5}
\begin{pmatrix}
-1&-2\\
-2&1\\
\end{pmatrix}
;\;\;\;
{\mathfrak D}(2,4)=\frac{1}{5}
\begin{pmatrix}
3&-2+2i\\
-2-2i&-3\\
\end{pmatrix};\;\;\;
{\mathfrak D}(3,4)=\frac{1}{10}
\begin{pmatrix}
3&-9\\
-9&-3\\
\end{pmatrix}\nonumber\\&&
{\mathfrak D}(1,2,4)=\frac{1}{5}
\begin{pmatrix}
-3&2-2i\\
2+2i&-2\\
\end{pmatrix};\;\;\;
{\mathfrak D}(1,3,4)=\frac{1}{10}
\begin{pmatrix}
-3&9\\
9&-7\\
\end{pmatrix}
\nonumber\\&&
{\mathfrak D}(2,3,4)=\frac{1}{10}
\begin{pmatrix}
-11&9-4i\\
9+4i&1\\
\end{pmatrix};\;\;\;
{\mathfrak D}(1,2,3,4)=\frac{1}{10}
\begin{pmatrix}
11&-9+4i\\
-9-4i&9\\
\end{pmatrix}
\end{eqnarray}
Then
\begin{eqnarray}\label{eq30}
\frac{1}{2}\sigma  (1)&=&\Pi (1)+\frac{1}{2}[{\mathfrak D} (1,2)+{\mathfrak D} (1,3)+{\mathfrak D} (1,4)]+\frac{1}{3}
[{\mathfrak D} (1,2,3)+{\mathfrak D} (1,2,4)+{\mathfrak D} (1,3,4)]\nonumber\\&+&\frac{1}{4}{\mathfrak D} (1,2,3,4)
\end{eqnarray}
and similarly for $\sigma (2), \sigma (3), \sigma (4)$.
Therefore
\begin{eqnarray}\label{sigma2}
&&\Pi(1)=\begin{pmatrix}
1&0\\
0&0\\
\end{pmatrix}\rightarrow \sigma  (1)=
\begin{pmatrix}
0.850&-0.150+0.066i\\
-0.150-0.066i&0.150\\
\end{pmatrix};\nonumber\\
&&\Pi(2)=\frac{1}{5}\begin{pmatrix}
1&-2i\\
2i&4\\
\end{pmatrix}\rightarrow \sigma  (2)=
\begin{pmatrix}
0.316&-0.150-0.200i\\
-0.150+0.200i&0.684\\
\end{pmatrix}\nonumber\\
&&\Pi(3)=\frac{1}{2}\begin{pmatrix}
1&1\\
1&1\\
\end{pmatrix}\rightarrow \sigma  (3)=
\begin{pmatrix}
0.516&0.183+0.067i\\
0.183-0.066i&0.484\\
\end{pmatrix}\nonumber;\\
&&\Pi(4)=\begin{pmatrix}
0&0\\
0&1\\
\end{pmatrix}\rightarrow \sigma  (4)=
\begin{pmatrix}
0.316&0.117+0.067i\\
0.117-0.066i&0.684\\
\end{pmatrix}
\end{eqnarray}
The resolution of the identity is
\begin{eqnarray}
\frac{1}{2}[\sigma  (1)+\sigma  (2)+\sigma  (3)+\sigma (4)]={\bf 1}.
\end{eqnarray}

\section{Non-independence and redundancy in the generalized bases }
\subsection{The coefficients ${s}_\theta (i)$ of Hermitian operators with respect to a generalized basis}
We consider a Hermitian operator $\theta $ and the $n$ real numbers
\begin{eqnarray}\label{47}
{s}_\theta (i)=\frac{d}{n}{\rm Tr}[\theta\sigma(i)];\;\;\;\sum _{i=1}^ns_\theta (i)={\rm Tr}(\theta).
\end{eqnarray}
Using the notation
\begin{eqnarray}
\theta _{\alpha \beta}=\bra{X;\alpha }\theta \ket{X;\beta};\;\;\;\sigma_{\alpha \beta }(i)=\bra {X;\alpha}\sigma(i)\ket {X;\beta}
;\;\;\;\alpha, \beta \in {\mathbb Z}(d),
\end{eqnarray}
we get
\begin{eqnarray}
s_\theta (i)=\frac{d}{n}\sum _{\alpha ,\beta}\theta _{\alpha \beta}\sigma_{\beta \alpha}(i).
\end{eqnarray}

We assume that the $n$ values of $s_\theta (i)$ are known, and the $d^2$ values of  $\theta _{\alpha \beta}$ are unknown.
Then this is a system of $n$ equations with $d^2$ unknowns. There are three cases:
\begin{itemize}
\item
If $n=d^2$, we can calculate $\theta _{\alpha \beta}$ (i.e., the operator $\theta$) from $s_\theta (i)$.
This is the case if we consider projectors $\Pi(i)$ associated to coherent states.
We have studied this case in \cite{V2017}.
\item
If $n>d^2$, and the values of $s_\theta (i)$ are accurate, the $n$ equations are compatible, and the system has an exact solution.
If the values of $s_\theta (i)$ are `noisy\rq{}, we can still find an `optimum solution\rq{}. All computer libraries can solve 
systems with more equations than unknowns, by minimizing the error, i.e., by minimizing the incompatibility between the equations.
\item
In the case $d< n<d^2$, we cannot calculate the $\theta _{\alpha \beta}$.
However, the information contained in $[s_\theta (1),...,s_\theta (n)]$ 
might be enough for certain physical conclusions.
In particular we show that change in the order of these numbers, might be linked to drastic physical changes in the system.
It is this case, with $d< n<d^2$, that we study in this paper.
\end{itemize}

\subsection{The ${s}_\rho (i)$ as pseudo-probabilities for density matrices }

If $\theta$ is a density matrix $\rho$, the $s_\rho (i)$ are results of measurements on $\rho$ with the
Hermitian operators $\sigma (i)$. A given $\sigma (i)$, is measurements with all its eigenprojectors
$\ket {{\mathfrak E}_\alpha (i)}\bra{{\mathfrak E}_\alpha (i)}$ (each of which gives a `yes-no\rq{} outcome), with weights its eigenvalues $e_\alpha (i)$:
\begin{eqnarray}\label{6n9}
s_\rho (i)=\frac{d}{n}{\rm Tr}[\rho\sigma(i)]=\frac{d}{n}\sum _{\alpha =1}^de_\alpha (i)\bra{{\mathfrak E}_\alpha (i)}\rho \ket {{\mathfrak E}_\alpha (i)}
\end{eqnarray}
Measurements with different $\sigma (i)$ are incompatible (they do not commute), 
and they need to be performed on different ensembles describing the same density matrix $\rho$.
The $n$ outcomes of such measurements are non-independent, but obey the relations
\begin{eqnarray}\label{478a}
0\le s_\rho (i)\le \frac{d}{n}{\mathfrak M}[\sigma _\rho (i)]< \frac{d}{n}<1;\;\;\;\sum _{i=1}^n s_\rho (i)=1.
\end{eqnarray}
Here  ${\mathfrak M}[\sigma _\rho (i)]$ is the maximum eigenvalue of the density matrix $\sigma (i)$.
The $s_\rho (i)\le \frac{d}{n}{\mathfrak M}[\sigma _\rho (i)]$ follows from Eq.(\ref{6n9}), if we replace all eigenvalues with the maximum eigenvalue.

We call the $s_\rho (i)$ pseudo-probabilities, where the `probabilities\rq{} indicates that they obey Eq.(\ref{478a}), and the `pseudo' indicates that they 
correspond to non-independent alternatives. Independent alternatives in the present context, correspond to orthonormal bases.
Since $s_\rho (i)< \frac{d}{n}$, the case
\begin{eqnarray}
&&s_\rho (i)=1\;\;{\rm if}\;\;i=i_0\nonumber\\
&&s_\rho (i)=0\;\;{\rm otherwise},
\end{eqnarray}
is not allowed for pseudo-probabilities. This  shows clearly the non-independence in the generalized bases.

\subsection{Use of Shannon entropy to quantify the non-independence and redundancy in generalized bases }

An entropic quantity \cite{E1,E2} that involves $n$ probabilities, takes values between $0$ and $\log n$.
We show that the entropy of our $n$ pseudo-probabilities, takes values between $(\log n -\log d)$ and $\log n$.
The lower bound is intimately related to the fact that $s_\rho (i)\le \frac{d}{n}$. 
\begin{definition}
The Shannon entropy
of a density matrix $\rho$ with respect to our generalizes bases, is given by:
\begin{eqnarray}
E_n (\rho)=-\sum _{i=1}^n s_\rho (i)\log [s_\rho (i)].
\end{eqnarray}
\end{definition}

\begin{proposition}
The Shannon entropy of a density matrix $\rho$ with respect to a generalized basis, is bounded as follows:
\begin{eqnarray}
\log n-\log d< E_n (\rho)\le \log n.
\end{eqnarray}
\end{proposition}

\begin{proof}
\begin{eqnarray}
E_n (\rho)&=&-\sum _{i=1}^n \left \{\frac{d}{n}{\rm Tr}[\rho \sigma(i)]\right \}\log \left \{\frac{d}{n}{\rm Tr}[\rho \sigma(i)]\right \}\nonumber\\&=&
-\log \left (\frac{d}{n}\right )\sum _{i=1}^n \left \{\frac{d}{n}{\rm Tr}[\rho \sigma(i)]\right \}-\frac{d}{n}
\sum _{i=1}^n{\rm Tr}[\rho \sigma(i)]\log {\rm Tr}[\rho \sigma(i)]
\end{eqnarray}
Taking into account Eq.(\ref{47}), and the fact that $0\le {\rm Tr}[\rho \sigma(i)]\le 1$, we get
\begin{eqnarray}
E_n (\rho)&=&
-\log \left (\frac{d}{n}\right )-\frac{d}{n}\sum _{i=1}^n{\rm Tr}[\rho \sigma(i)]\log {\rm Tr}[\rho \sigma(i)]> \log n-\log d.
\end{eqnarray}
We note here that $0\le {\rm Tr}[\rho \sigma(i)]\le {\mathfrak M}[\sigma (i)]<1$ (Eq.(\ref{478a})), 
and therefore the $\sum {\rm Tr}[\rho \sigma(i)]\log {\rm Tr}[\rho \sigma(i)]$ is non-zero.

For the upper bound, we point out that $E_n (\rho)$ involves $n$ probabilities, and therefore $\log n$ is an upper bound.
\end{proof}

\begin{example}
\mbox{}
\begin{itemize}
\item
If $\rho =\frac{1}{d}{\bf 1}$ then 
\begin{eqnarray}
s_\rho (i)=\frac{1}{n};\;\;\;E_n \left (\frac{1}{d}{\bf 1}\right )=\log n.
\end{eqnarray}
\item
If $\rho=\ket{X;\alpha}\bra{X;\alpha}$, then
\begin{eqnarray}
&&s_\rho (i)=\frac{d}{n}\sigma_{\alpha \alpha }(i)\nonumber\\
&&E_n (\ket{X;\alpha}\bra{X;\alpha})=(\log n-\log d)-\frac{d}{n}\sum _{i=1}^n \left [\sigma_{\alpha \alpha }(i)\right ]
\log \left [\sigma_{\alpha \alpha }(i)\right ]
\end{eqnarray}
For $\rho=\ket{X;0}\bra{X;0}$ and with the generalized basis in Eq.(\ref{sigma}), we get
\begin{eqnarray}
&&\sigma _{00}(1)=0.825;\;\;\;\sigma _{00}(2)=0.225;\;\;\;\sigma _{00}(3)=0.450\nonumber\\
&&E_3(\ket{X;0}\bra{X;0})=\log3-\log2+0.569=0.974
\end{eqnarray}
For $\rho=\ket{X;0}\bra{X;0}$ and with the generalized basis in Eq.(\ref{sigma2}), we get
\begin{eqnarray}
&&\sigma _{00}(1)=0.850;\;\;\;\sigma _{00}(2)=0.316;\;\;\;\sigma _{00}(3)=0.516;\;\;\;\sigma _{00}(4)=0.316\nonumber\\
&&E_4(\ket{X;0}\bra{X;0})=\log4-\log2+0.603=1.296
\end{eqnarray}
We use the base $e$ for logarithms, and the results are in nats. 
\end{itemize}
\end{example}
We have seen above, that the upper bound $\log n$ in the set of $\{E_n(\rho)\}$, 
is reached with the density matrix $\rho =\frac{1}{d}{\bf 1}$.
We have also seen that $\log n-\log d$ is a lower bound,  
but it is an open question what is the infimum.
We call the 
\begin{eqnarray}
{\mathfrak R}=\log n-\log d=\log ({\cal R}+1)
\end{eqnarray}
entropic redundancy index. It plays a complementary role to the redundancy index ${\cal R}$ in Eq.(\ref{red}).

We have shown that the Shannon entropies $E_n (\rho)$ take values in the interval between ${\mathfrak R}$ and ${\mathfrak R}+\log d$, which has length $\log d$, for any $n$. In the standard Shannon entropy with respect to an orthonormal basis, ${\mathfrak R}=0$.

\section{Representation of vectors in the generalized basis}

An arbitrary normalized vector in $H_d$ can now be expanded in terms of $n> d$ component vectors, as
\begin{eqnarray}\label{6bm}
\ket {V}=\sum _{i=1}^n \ket{V(i)};\;\;\;\ket{V(i)}=\frac{d}{n}\sigma (i) \ket{V}.
\end{eqnarray}
The scalar product is given by
\begin{eqnarray}
&&\langle V\ket {U}=\sum _{i,j}\bra{V}g(i,j) \ket{U};\;\;\;g(i,j)=\frac{d^2}{n^2}\sigma (i)\sigma(j)\nonumber\\
&&\sum_{i,j}g(i,j)={\bf 1};\;\;\;[g(i,j)]^{\dagger}=g(j,i).
\end{eqnarray}
The `metric\rq{} $g(i,j)$ consists of $n^2$ matrices, each of which is a $d\times d$ matrix.

We express the density matrices $\sigma(i)$ in terms of their eigenvalues (probabilities) $p_\alpha (i)$ and their eigenvectors 
$\ket{{\mathfrak E}_\alpha (i)}$, as:
\begin{eqnarray}
&&\sigma (i)=\sum _{\alpha =1}^d p_\alpha (i)\ket {{\mathfrak E}_\alpha (i)}\bra{{\mathfrak E}_\alpha (i)};\;\;\;\sum _{\alpha =1}^dp_\alpha (i)=1;\;\;\;\sum _{\alpha =1}^d\ket {{\mathfrak E}_\alpha (i)}\bra{{\mathfrak E}_\alpha (i)}={\bf 1}\nonumber\\
&&\frac{d}{n}\sum _{i=1}^n\sum _{\alpha =1}^d p_\alpha (i)\ket {{\mathfrak E}_\alpha (i)}\bra{{\mathfrak E}_\alpha (i)}={\bf 1}.
\end{eqnarray}
Our formalism renormalizes each projector $\ket{i}\bra{i}$ into the density matrix $\sigma (i)$ 
which can be viewed as a set
of orthonormal bases $\ket{{\mathfrak E}_\alpha (i)}$ with probabilities $p_\alpha (i)$ attached to them.

\begin{example}

In $H_2$ we consider the vector
\begin{eqnarray}\label{44}
\ket{V}=\frac{1}{\sqrt{15}}
\begin{pmatrix}
1+2i\\
3-i\\
\end{pmatrix}.
\end{eqnarray}
We also consider the matrices $\sigma (1)$, $\sigma (2)$, $\sigma (3)$, in Eq.(\ref{sigma}), and
using the resolution of the identity in Eq.(\ref{aw4}) we expand this vector as
\begin{eqnarray}
\ket{V}=\frac{2}{3}[\sigma  (1)\ket{V}+\sigma  (2)\ket{V}+\sigma  (3)\ket{V}]=
\begin{pmatrix}
0.094+0.357i\\
0.103-0.090i\\
\end{pmatrix}+
\begin{pmatrix}
-0.060-0.004i\\
0.309-0.142i\\
\end{pmatrix}+
\begin{pmatrix}
0.223+0.163i\\
0.361-0.025i\\
\end{pmatrix}.
\end{eqnarray}
There is redundancy and `duplication\rq{} in this approach, which is precisely the merit for using it.
Errors due to noise in some of these components are compensated by the other components, and the overall error is small, as discussed below.
\end{example}

\subsection{Robustness of the representation in the presence of noise}

We add noise to the $n$ components of the vector $\ket {V}$ in Eq. (\ref{6bm}), and we get the vector:
\begin{eqnarray}\label{478}
\ket {W}= \frac{d}{n}\sum _{i=1}^n(1+{\mathfrak N}_i)\sigma (i) \ket{V}.
\end{eqnarray}
Here ${\mathfrak N}_i$ are $n$ independent real random numbers, uniformly distributed in the interval $[-\mu, \mu]$ (in all numerical calculations in this paper $\mu=0.5$).

As a measure of the error we calculate the number
\begin{eqnarray}\label{error}
&&\epsilon =||\ket {W}-\ket{V}||=\sqrt {\epsilon _D+\epsilon _{ND}}\nonumber\\
&&\epsilon_D=\sum _{i}{\mathfrak N}_i^2{\bra{V}g(i,i)\ket{V}}\nonumber\\
&&\epsilon _{ND}=\sum _{i\ne j}{\mathfrak N}_i{\mathfrak N}_j{\bra{V}g(i,j)\ket{V}}.
\end{eqnarray}
$\epsilon_D$ contains the diagonal terms which are positive numbers, and $\epsilon_{ND}$ contains the non-diagonal terms which might be negative.

For comparison, we also expand the same vector in the orthonormal basis of position states, as
\begin{eqnarray}\label{E1}
\ket{V}=\sum _{\alpha=1}^dV(\alpha)\ket{X;\alpha};\;\;\;V(\alpha)=\bra{X;\alpha}V\rangle,
\end{eqnarray}
We then add noise in these $d$ components as follows:
\begin{eqnarray}\label{E2}
\ket {W _{\rm orth}} =\sum _{\alpha=1}^d[1+{\mathfrak N}_\alpha]V(\alpha)\ket{X;\alpha}
=\ket{V}+\sum _{\alpha=1}^d{\mathfrak N}_\alpha V(\alpha)\ket{X;\alpha}.
\end{eqnarray}
Here ${\mathfrak N}_\alpha$ are $d$ independent real random numbers, uniformly distributed in the interval $[-\mu, \mu]$.

As a measure of the error in this case, we calculate the number
\begin{eqnarray}\label{error2}
\epsilon _{\rm orth} ={||\ket{W_{\rm orth}} -\ket {V}||}=
\left [\sum _{\alpha}{\mathfrak N}_\alpha^2|V(\alpha)|^2\right ]^{1/2}.
\end{eqnarray}
Here we only have diagonal terms which are positive numbers.
Therefore we expect that in general the error $\epsilon$ will be smaller than the error $\epsilon _{\rm orth}$.
Numerical results below confirm that this is the case.

\begin{example}
In $H_2$ we consider the vector of Eq.(\ref{44}).
We used the three density matrices $\sigma  (1)$, $\sigma  (2)$, $\sigma  (3)$, in Eq.(\ref{sigma}) 
(which are renormalizations of the three vectors in Eq.(\ref{set1})) as a generalized basis.
Using three independent random numbers,
we calculated the errors in Eq.(\ref{error}), and we called them $\epsilon_{3D}$, $\epsilon_{3ND}$ and $\epsilon _3$.
We repeated the calculation five times (with different sets of random numbers) and found the errors given in table \ref{t1}.

We also used the four density matrices $\sigma  (1)$, $\sigma  (2)$, $\sigma  (3)$, $\sigma (4)$,  in Eq.(\ref{sigma2}),
(which are renormalizations of the four vectors in Eq.(\ref{set2})) as a generalized basis.
Using Eq.(\ref{478}) with four independent random numbers,
we calculated the errors of Eq.(\ref{error}), and we called them  $\epsilon_{4D}$, $\epsilon_{4ND}$ and $\epsilon _4$.
Results in this case are also given in table \ref{t1}.

Furthermore, we used the orthonormal basis in Eq.(\ref{E1}), and added noise in the two components as in Eq.(\ref{E2}),
using two independent random numbers.
We then calculated the error $\epsilon _{\rm orth}$ of Eq.(\ref{error2}), and give the results in table \ref{t1}.

The results show that the generalized bases of the density matrices $\sigma  (i)$, 
lead to smaller error than the orthonormal bases.
In some cases, the non-diagonal parts of the error $\epsilon_{3ND}$, $\epsilon_{4ND}$, are negative, and this contributes to the reduction of the error.
\end{example}

\section{Use of generalized bases to detect physical changes in the presence of noise}
\subsection{Location indices of a Hermitian operator}
\begin{definition}
Let $\theta (\lambda)$ be a Hermitian operator, e.g. a Hamiltonian that depends on a coupling parameter $\lambda$.
Also let  $s_\theta (i |\lambda)$ be the $n$ coefficients defined in Eq(\ref{47}) (which are here functions of $\lambda$). 
We order the $s_\theta (1 |\lambda),...,s_\theta (n|\lambda)$ as
\begin{eqnarray}\label{35}
s_\theta (i_1 |\lambda)\ge s_\theta (i_2|\lambda)\ge ...\ge s_\theta (i _n|\lambda). 
\end{eqnarray}
The location index of $\theta (\lambda)$, 
with respect to $\{\sigma (i)\}$, is the $n$-tuple
\begin{eqnarray}
{\cal L}[\theta(\lambda)]=(i_1,...,i_n)\in {\mathfrak T}.
\end{eqnarray}
Here ${\mathfrak T}$ is the set of the $n!$ permutations of the $n$ labels $i$, of $s_\theta (i |\lambda)$. 
\end{definition}

The ${\cal L}[\theta(\lambda)]$ indicates the position of $\theta (\lambda)$ with respect to the generalized basis of
$\{\sigma (i)\}$.
$\theta (\lambda)$ is more close to $\sigma (i_1)$ (because $s_\theta (i_1 |\lambda)$ is the largest),
less close to $\sigma (i_2)$, even less close to $\sigma (i_3)$, etc.

In ref\cite{V2017}, we used this concept with projectors $\Pi (i)$ related to coherent states
which are linked to the familiar concept of phase space, and then the ${\cal L}[\theta(\lambda)]$ (with $n=d^2$)
locates the operator $\theta (\lambda)$ in phase space.
Here the physical interpretation of ${\cal L}[\theta(\lambda)]$ is more abstract, 
because the $\Pi (i)$ are arbitrarily chosen.
Nevertheless, the ${\cal L}[\theta(\lambda)]$ describes the position of $\theta (\lambda)$ with respect to $\{\sigma (i)\}$,
which  resolve the identity.

Operators $\theta (\lambda)$ for which the $n$ values $s_\theta (i |\lambda)$ 
(with $i=1,...,n$) are different from each other (i.e., there is no equality in Eq.(\ref{35}))
are described by only one permutation
$(i_1,...,i_n)$. This motivates the following definition.
\begin{definition}
For a given set $\Theta=\{\theta (\lambda)\;|\;\lambda \in [a,b]\}$, its subset ${\widetilde \Theta}=\{\theta (\lambda)\;|\;\lambda \in I\subseteq [a,b]\}$ contains 
all $\theta (\lambda)$ for which the $n$ values $s_\theta (i |\lambda)$ 
(with fixed $\lambda$ and $i=1,...,n$) are different from each other.
The interval $I$ excludes all values of $\lambda$ for which there are some equalities in Eq.(\ref{35}).
\end{definition}

\begin{proposition}
Within the set ${\widetilde \Theta}$, we say that $\theta (\lambda _1)$ and $\theta (\lambda _2)$
are comonotonic, and denote it as $\theta (\lambda _1)\sim \theta (\lambda _2)$, if ${\cal L}[\theta(\lambda _1)]={\cal L}[\theta(\lambda _2)]$.
Then $\sim$ is an equivalence relation and ${\widetilde \Theta}$ is partitioned into equivalence classes,
each of which contains operators which are comonotonic to each other.
\end{proposition}
\begin{proof}
The proofs of reflexibity ($\theta (\lambda _1)\sim \theta (\lambda _1)$), and symmetry
(if $\theta (\lambda _1)\sim \theta (\lambda _2)$ then $\theta (\lambda _2)\sim \theta (\lambda _1)$), are trivial.
Transitivity holds within ${\widetilde \Theta}$.
Indeed if ${\cal L}[\theta(\lambda _1)]={\cal L}[\theta(\lambda _2)]$ and ${\cal L}[\theta(\lambda _2)]={\cal L}[\theta(\lambda _3)]$
then ${\cal L}[\theta(\lambda _1)]={\cal L}[\theta(\lambda _3)]$.
It is important for the proof that only one permutation corresponds to a given $\theta (\lambda)$.
For this reason, transitivity does not hold within ${\Theta}$, in general.
\end{proof}

\begin{definition}
If all $\theta (\lambda)$ in the set  $\{\theta (\lambda)\;|\;\lambda \in (c_1,c_2)\}$ are comonotonic to each other, the $R=(c_1,c_2)$ is called comonotonicity interval (with respect to the operators $\theta (\lambda)$).
The points in the set $[a,b]\setminus I$ are crossing points from one comonotonicity region to another.
\end{definition}

In this paper we show with examples, that comonotonic operators are physically similar operators.
As $\lambda$ varies within a comonotonicity interval, we get mild physical changes in the system.  
The crossing points from one comonotonicity interval to another, 
might be related with drastic physical changes in the system.
In the example below, this involves abrupt change in the ground state of the system.

\subsection{Ground state of a physical system}
In the Hilbert space $H_2$ we consider a system with Hamiltonian which is described with the matrix
\begin{eqnarray}\label{hami}
\theta (\lambda)=\begin{pmatrix}
1 &0\\
0&1\\
\end{pmatrix}+\lambda
\begin{pmatrix}
0 &1+i\\
1-i&0\\
\end{pmatrix}.
\end{eqnarray}
This two-dimensional system is used many times as an approximation to an infinite-dimensional system, 
where due to low energy the system is practically in the subspace of the lowest two states.
Many of the experimentally available qubits are of this type (e.g., the superconducting qubits).

We will study changes to the ground state of the system as the coupling parameter $\lambda $ varies, from negative to positive values.
We will show that at $\lambda=0$ the ground state of the system changes abruptly from one vector, to another one which is orthogonal to it. 

A method is practically useful if it is robust in the presence of noise.
If we add to the `real\rq{} values of the parameters a small amount of noise (due to experimental and other errors), the results should not change much. 
In order to study this we consider the `noisy Hamiltonian\rq{}
\begin{eqnarray}\label{54}
\phi (\lambda)=\begin{pmatrix}
1+{\mathfrak N}_1 &0\\
0&1+{\mathfrak N}_2\\
\end{pmatrix}
+\lambda
\begin{pmatrix}
0 &1+i\\
1-i&0\\
\end{pmatrix}.
\end{eqnarray}
For simplicity, we add noise only to the `free part\rq{} of the Hamiltonian, with the independent random numbers 
${\mathfrak N}_1, {\mathfrak N}_2$, which are uniformly distributed in the interval $[-\mu, \mu]$.
$\phi (\lambda)$ is an approximation to the `real Hamiltonian\rq{} $\theta (\lambda)$.
We will show that the ground state of  $\phi (\lambda)$ changes
rapidly but smoothly, within a small region of $\lambda$, around $\lambda =0$ with width 
$|{\mathfrak N}_1-{\mathfrak N}_2|$ .
The abrupt change of the ground state of ${\theta} (\lambda)$ at $\lambda=0$,
becomes a rapid but smooth change of the ground state of $\phi (\lambda)$, within a small region around $\lambda =0$.

Our method based on generalized bases is complementary to the calculation of eigenvalues and eigenvectors, 
and is robust in the presence of noise, because of the redundancy which is inherent in it.
For the noiseless Hamiltonian $\theta (\lambda)$, there are two comonotonicity regions $(-\infty, 0)$ and $(0, \infty)$, and the point
$\lambda=0$ is a crossing point from the first comonotonicity region to the second one.
For the noisy Hamiltonian $\phi (\lambda)$, there are more crossing points near $\lambda =0$, which indicate 
that drastic physical changes occur in that region.
There are no crossing points far from $\lambda =0$, and this reflects the fact that only mild physical changes occur there.

\subsubsection{Noiseless Hamiltonians at zero temperature in a generalized basis}

The eigenvalues (energy levels) and eigenvectors of the `noiseless Hamiltonian\rq{} $\theta (\lambda)$, are
\begin{eqnarray}\label{50}
&&e_1(\lambda)=1+\lambda\sqrt{2};\;\;\;
\ket{{\mathfrak e}_1}=\frac{1}{2}\begin{pmatrix}
(1+i)\\
\sqrt {2}\\
\end{pmatrix}\nonumber\\
&&e_2(\lambda)=1-\lambda \sqrt{2};\;\;\;
\ket{{\mathfrak e}_2}=\frac{1}{2}\begin{pmatrix}
-(1+i)\\
\sqrt {2}\\
\end{pmatrix};\;\;\;\langle {\mathfrak e}_1\ket{{\mathfrak e}_2}=0.
\end{eqnarray}
For $\lambda <0$, the $\ket{{\mathfrak e}_1}$ is the 
ground state of the system,
while for $\lambda >0$, the $\ket{{\mathfrak e}_2}$ is the 
ground state of the system. At $\lambda =0$ the two eigenvalues become equal to each other, and the ground state
changes abruptly from $\ket{{\mathfrak e}_1}$ for $\lambda <0$, to 
$\ket{{\mathfrak e}_2}$ (which is orthogonal to $\ket{{\mathfrak e}_1}$) for
$\lambda >0$. 

We next use the generalized bases studied in this paper.
We first use the density matrices in Eq.(\ref{sigma}) we find that the $s_\theta(i|\lambda)$
are
\begin{eqnarray}
&&s_\theta(1|\lambda)=\frac{2}{3}[1-0.050 \lambda];\;\;\;s_\theta(2|\lambda)=\frac{2}{3}[1-0.650 \lambda];\;\;\;s_\theta(3|\lambda)=\frac{2}{3}[1
+0.700\lambda]\nonumber\\
&&s_\theta(1|\lambda)+s_\theta(2|\lambda)+s_\theta(3|\lambda)=2.
\end{eqnarray}
Therefore we have two comonotonicity regions (which we give together with the corresponding location indices for $\theta (\lambda)$):
\begin{eqnarray}
&&R_1=\left (-\infty,0\right);\;\;\;{\cal L}[\theta(\lambda)]=(2,1,3)\nonumber\\
&&R_2=\left (0,\infty\right);\;\;\;{\cal L}[\theta(\lambda)]=(3,1,2)
\end{eqnarray}
At $\lambda =0$ we pass from the first comonotonicity region to the second one, and this is associated with drastic physical changes in the ground state of the system.  

We also use the density matrices in Eq.(\ref{sigma2}) we find that the $s_\theta(i|\lambda)$ are
\begin{eqnarray}
&&s_\theta(1|\lambda)=\frac{1}{2}[1-0.168 \lambda];\;\;\;s_\theta(2|\lambda)=\frac{1}{2}[1-0.700 \lambda]\nonumber\\
&&s_\theta(3|\lambda)=\frac{1}{2}[1+0.500\lambda]
;\;\;\;s_\theta(4|\lambda)=\frac{1}{2}[1+0.368\lambda]\nonumber\\
&&s_\theta(1|\lambda)+s_\theta(2|\lambda)+s_\theta(3|\lambda)+s_\theta(3|\lambda)=2.
\end{eqnarray}
Therefore we have two comonotonicity regions:
\begin{eqnarray}
&&R_1=\left (-\infty,0\right);\;\;\;{\cal L}[\theta(\lambda)]=(3,4,1,2)\nonumber\\
&&R_2=\left (0,\infty\right);\;\;\;{\cal L}[\theta(\lambda)]=(2,1,4,3)
\end{eqnarray}
It is seen that with this generalized basis also, we arrive at the same conclusions.  
Two different generalized bases lead to the same conclusion as the method of eigenvectors and eigenvalues.

\subsubsection{Noiseless Hamiltonians at finite temperature in a generalized basis}

Let 
\begin{eqnarray}
{\cal E}=\exp [-\beta \theta (\lambda)];\;\;\;{s}_{\cal E} (i)=\frac{d}{n}{\rm Tr}[{\cal E}\sigma(i)],
\end{eqnarray}
where $\beta$ is the inverse temperature.
Then the partition function is
\begin{eqnarray}
Z={\rm Tr}{\cal E}=\sum _{i=1}^ns_{\cal E} (i).
\end{eqnarray}
For the Hamiltonian $\theta (\lambda)$, we  get
\begin{eqnarray}
{\cal E}=e^{-\beta}
\begin{pmatrix}
\cosh (\beta \lambda \sqrt{2})&-\frac{1+i}{\sqrt{2}}\sinh (\beta \lambda \sqrt{2})\\
-\frac{1-i}{\sqrt{2}}\sinh (\beta \lambda \sqrt{2})&\cosh (\beta \lambda \sqrt{2})\\
\end{pmatrix}
\end{eqnarray}
We use the density matrices in Eq.(\ref{sigma}), and we find that the $s_{\cal E}(i|\lambda)$
are
\begin{eqnarray}
&&s_{\cal E}(1|\lambda)=\frac{2}{3}e^{-\beta}[\cosh (\beta \lambda \sqrt{2})+0.035\sinh (\beta \lambda \sqrt{2})]\nonumber\\
&&s_{\cal E}(2|\lambda)=\frac{2}{3}e^{-\beta}[\cosh (\beta \lambda \sqrt{2})+0.459\sinh (\beta \lambda \sqrt{2})]\nonumber\\
&&s_{\cal E}(3|\lambda)=\frac{2}{3}e^{-\beta}[\cosh (\beta \lambda \sqrt{2})-0.494\sinh (\beta \lambda \sqrt{2})].
\end{eqnarray}
We also use the density matrices in Eq.(\ref{sigma2}), and we find that the $s_{\cal E}(i|\lambda)$
are
\begin{eqnarray}
&&s_{\cal E}(1|\lambda)=\frac{1}{2}e^{-\beta}[\cosh (\beta \lambda \sqrt{2})+0.118\sinh (\beta \lambda \sqrt{2})]\nonumber\\
&&s_{\cal E}(2|\lambda)=\frac{1}{2}e^{-\beta}[\cosh (\beta \lambda \sqrt{2})+0.494\sinh (\beta \lambda \sqrt{2})]\nonumber\\
&&s_{\cal E}(3|\lambda)=\frac{1}{2}e^{-\beta}[\cosh (\beta \lambda \sqrt{2})-0.352\sinh (\beta \lambda \sqrt{2})]\nonumber\\
&&s_{\cal E}(4|\lambda)=\frac{1}{2}e^{-\beta}[\cosh (\beta \lambda \sqrt{2})-0.260\sinh (\beta \lambda \sqrt{2})].
\end{eqnarray}
In calculations that involve the partition function, we can use a generalized basis and the $s_{\cal E} (i|\lambda)$, instead of an orthonormal basis. 
The merit of this, is robustness of the results in the presence of noise, as we show with examples below.

We note that the partition function is
\begin{eqnarray}
Z=2e^{-\beta}\cosh (\beta \lambda \sqrt{2}),
\end{eqnarray}
and from this we find the average energy
\begin{eqnarray}
<e(\lambda)>=-\frac{1}{Z}\frac{\partial Z}{\partial \beta}=1-\lambda \sqrt{2}\tanh (\beta \lambda \sqrt{2}).
\end{eqnarray}
It is seen that at low temperatures ($\beta \rightarrow \infty$), 
\begin{eqnarray}
&&\lambda >0\;\;\rightarrow\;\;<e(\lambda)>\approx 1-\lambda \sqrt{2};\nonumber\\
&&\lambda <0\;\;\rightarrow\;\;<e(\lambda)>\approx 1+\lambda \sqrt{2}.
\end{eqnarray}
This is consistent with the result in Eq.(\ref{50}), at zero temperatures.

\subsubsection{Hamiltonians with noise at zero temperature: eigenvalues approach}

The eigenvalues and eigenvectors of the `noisy Hamiltonian\rq{} $\phi (\lambda)$, are
\begin{eqnarray}\label{321}
&&e_A(\lambda)=1+S-\sqrt{D^2+2\lambda ^2}
;\;\;\;S=\frac{{\mathfrak N}_1+{\mathfrak N}_2}{2};\;\;\;D=\frac{{\mathfrak N}_1-{\mathfrak N}_2}{2}\nonumber\\
&&e_B(\lambda)=1+S+\sqrt{D^2+2\lambda ^2}
\end{eqnarray}
It is convenient  to replace the random numbers 
${\mathfrak N}_1, {\mathfrak N}_2$, with the $S,D$ which are also independent random numbers.
The corresponding eigenvectors (not normalized) are given by
\begin{eqnarray}
&&\ket{{\mathfrak e}_A(\lambda)}=\begin{pmatrix}
-\frac{\lambda}{|\lambda|}(1+i)\\
\frac{D}{|\lambda|}+\sqrt{2+\left (\frac {D}{|\lambda|}\right )^2}\\
\end{pmatrix}\nonumber\\
&&\ket{{\mathfrak e}_B(\lambda)}=\begin{pmatrix}
-\frac{\lambda}{|\lambda|}(1+i)\\
\frac{D}{|\lambda|}-\sqrt{2+\left (\frac {D}{|\lambda|}\right )^2}
\end{pmatrix};\;\;\;\langle {\mathfrak e}_A(\lambda)\ket{{\mathfrak e}_B(\lambda)}=0.
\end{eqnarray}
The eigenvectors depend on the sign of $\lambda$ and on the value of $\frac {D}{|\lambda|}$.
The lowest eigenvalue is $e_A(\lambda)$ and the corresponding eigenvector $\ket{{\mathfrak e}_A(\lambda)}$.

For small values of $\frac {D}{|\lambda|}$ the ground state $\ket{{\mathfrak e}_A(\lambda)}$ can be written as
\begin{eqnarray}
\ket{{\mathfrak e}_A(\lambda)}=\begin{pmatrix}
-\frac{\lambda}{|\lambda|}(1+i)\\
\sqrt {2}+\frac{D}{|\lambda|}+\frac{\sqrt{2}}{4}\left (\frac{D}{|\lambda|}\right ) ^2-...\\
\end{pmatrix}
\end{eqnarray}
In this case for $\lambda <0$, we get $\ket{{\mathfrak e}_A(\lambda)}\approx \ket{{\mathfrak e}_1}$, and for  $\lambda >0$, 
we get $\ket{{\mathfrak e}_A(\lambda)}\approx \ket{{\mathfrak e}_2}$.
It is seen that when the noise parameter $D$ is much smaller than the coupling parameter, we recover the results of the noiseless case, discussed earlier.

Without loss of generality, we assume that $D\ge 0$.
The physically interesting and practically useful case, is to assume a fixed noise parameter $D$, 
and study the ground state as $\lambda$ varies within the region $(-D, D)$, and in particular very close to $0$.
This is the limit of large values of  $\frac {D}{|\lambda|}$.
We compare $\ket{{\mathfrak e}_A(-|\lambda|)}$ with $\ket{{\mathfrak e}_A(|\lambda|)}$, and see to what extend they are orthogonal as in the noiseless case.
In particular we calculate the overlap 
\begin{eqnarray}\label{r}
r(|\lambda|)=\frac{\bra{{\mathfrak e}_A(-|\lambda|)}{\mathfrak e}_A(|\lambda|)\rangle}
{\sqrt {\bra{{\mathfrak e}_A(-|\lambda|)}{\mathfrak e}_A(-|\lambda|)\rangle\bra{{\mathfrak e}_A(|\lambda|)}{\mathfrak e}_A(|\lambda|)\rangle}}=
\frac{-2+A^2}{2+A^2};\;\;\;A=\frac{D}{|\lambda|}+\sqrt {2+\left (\frac {D}{|\lambda|}\right )^2}.
\end{eqnarray}
For fixed $D$ and when $\lambda$ is close to zero, the  $\frac {D}{|\lambda|}$ is large, and the $r(|\lambda|)$ is close to $1$.
It is seen that as $\lambda $ changes from negative to positive values, the $\ket{{\mathfrak e}_A(-|\lambda|)}$ changes 
quickly but smoothly to $\ket{{\mathfrak e}_A(|\lambda|)}$
(the angle between these two vectors is small and decreases gradually as $|\lambda|$ goes near $0$).
There are no discontinuities, in the sense that for
any given value of $r(|\lambda|)$, we can find the value of $\frac{D}{|\lambda|}$ which leads to it.
Therefore in the presence of noise, the method of the eigenvalues and eigenvectors cannot find the abrupt change in the ground state of the 
`real system\rq{}, at $\lambda =0$. 
Instead, it finds rapid but smooth changes of the ground state within the small interval $(-D,D)$, and slow changes in the large region outside it. 

Above we worked with the eigenvalues $e_A(\lambda), e_B(\lambda)$ which are random numbers. An alternative approximative approach, will be to work with their expectation values.
We assume that the average value of the random variables $S, D$ is $0$, and that the 
standard deviation of $D$ is $\sigma$.
If $g(D)$ is a function of $D$, then its expectation value $E[g(D)]$ is given by (e.g., Eq.(5-61) in \cite{PA})
\begin{eqnarray}
E[g(D)]=g(0)+g^{\prime \prime}(0)\frac{\sigma^2}{2}+...
\end{eqnarray}
If we ignore the higher moments, we get
\begin{eqnarray}
E[\sqrt{D^2+2\lambda ^2}]\approx |\lambda|\sqrt{2}+\frac{\sigma ^2}{2^{3/2}|\lambda|}.
\end{eqnarray}
Therefore
\begin{eqnarray}\label{34a}
E[e_A(\lambda)]\approx 1-|\lambda|\sqrt{2}-\frac{\sigma ^2}{2^{3/2}|\lambda|};\;\;\;
E[e_B(\lambda)]\approx 1+|\lambda|\sqrt{2}+\frac{\sigma ^2}{2^{3/2}|\lambda|}.
\end{eqnarray}
This approach also shows that the ground state energy (averaged over noise), is
$E[e_A(\lambda)]$, and as we go from negative to positive values of $\lambda$, 
the ground state changes from  $\ket{{\mathfrak e}_A(-|\lambda|)}$ to $\ket{{\mathfrak e}_A(|\lambda|)}$.
As we explained above (using Eq.(\ref{r})) this is a smooth but quick change of the ground state.

\subsubsection{Hamiltonians with noise at zero temperature: generalized bases approach}
We next use the generalized bases studied in this paper.
We first use the density matrices in Eq.(\ref{sigma}) we find that the $s_\theta(i|\lambda)$
are:
\begin{eqnarray}
&&s_\theta(1|\lambda)=\frac{2}{3}[1+0.825{\mathfrak N}_1+0.175{\mathfrak N}_2-0.050 \lambda]=
\frac{2}{3}[1+{\mathfrak N}_1-0.350D-0.050 \lambda]
\nonumber\\
&&s_\theta(2|\lambda)=\frac{2}{3}[1+0.225{\mathfrak N}_1+0.775{\mathfrak N}_2-0.650 \lambda]=
\frac{2}{3}[1+{\mathfrak N}_1-1.550D-0.650 \lambda]\nonumber\\
&&s_\theta(3|\lambda)=\frac{2}{3}[1+0.450{\mathfrak N}_1+0.550{\mathfrak N}_2+0.700\lambda]=
\frac{2}{3}[1+{\mathfrak N}_1-1.100D+0.700\lambda]\nonumber\\
&&s_\theta(1|\lambda)+s_\theta(2|\lambda)+s_\theta(3|\lambda)=2+{\mathfrak N}_1+{\mathfrak N}_2.
\end{eqnarray}
We assume that $D>0$ and  we find the following comonotonicity intervals (which we give together with the corresponding location indices of $\theta(\lambda)$):
\begin{eqnarray}
&&R_1=\left (-\infty,-2D\right);\;\;\;{\cal L}[\theta(\lambda)]=(2,1,3)\nonumber\\
&&R_2=\left (-2D,-0.333D\right);\;\;\;{\cal L}[\theta(\lambda)]=(1,2,3)\nonumber\\
&&R_3=\left (-0.333D,1.153D\right);\;\;\;{\cal L}[\theta(\lambda)]=(1,3,2)\nonumber\\
&&R_4=\left (1.153D,\infty\right);\;\;\;{\cal L}[\theta(\lambda)]=(3,1,2)
\end{eqnarray}
There are three crossing points near $\lambda =0$ (at $-2D, -0,333D, 1.153D$), which indicate
that drastic physical changes occur in that region.
There are no crossing points far from $\lambda =0$, and this indicates that only mild physical changes occur there.
We note that if we average over the random variable $D$, then we get two comonotonicity regions
$(-\infty, 0)$ and $(0, \infty)$ as in the noiseless case.

We also use the density matrices in Eq.(\ref{sigma2}), and we find that the $s_\theta(i|\lambda)$
are:
\begin{eqnarray}
&&s_\theta(1|\lambda)=\frac{1}{2}[1+0.850{\mathfrak N}_1+0.150{\mathfrak N}_2-0.168 \lambda]
=\frac{1}{2}[1+{\mathfrak N}_1-0.300D-0.168 \lambda]\nonumber\\
&&s_\theta(2|\lambda)=\frac{1}{2}[1+0.316{\mathfrak N}_1+0.684{\mathfrak N}_2-0.700 \lambda]
=\frac{1}{2}[1+{\mathfrak N}_1-1.368D-0.700 \lambda]\nonumber\\
&&s_\theta(3|\lambda)=\frac{1}{2}[1+0.516{\mathfrak N}_1+0.484{\mathfrak N}_2+0.500 \lambda]
=\frac{1}{2}[1+{\mathfrak N}_1-0.968D+0.500 \lambda]\nonumber\\
&&s_\theta(4|\lambda)=\frac{1}{2}[1+0.316{\mathfrak N}_1+0.684{\mathfrak N}_2+0.368 \lambda]
=\frac{1}{2}[1+{\mathfrak N}_1-1.368D+0.368\lambda]\nonumber\\
&&s_\theta(1|\lambda)+s_\theta(2|\lambda)+s_\theta(3|\lambda)+s_\theta(3|\lambda)=2+{\mathfrak N}_1+{\mathfrak N}_2.
\end{eqnarray}

We assume that $D>0$ and  we find the following comonotonicity intervals:
\begin{eqnarray}
&&R_1=\left (-\infty,-3D\right);\;\;\;{\cal L}[\theta(\lambda)]=(2,1,4,3)\nonumber\\
&&R_2=\left (-3D,-2D\right);\;\;\;{\cal L}[\theta(\lambda)]=(2,1,3,4)\nonumber\\
&&R_3=\left (-2D,-0.33D\right);\;\;\;{\cal L}[\theta(\lambda)]=(1,2,3,4)\nonumber\\
&&R_4=\left (-0.33D,0\right);\;\;\;{\cal L}[\theta(\lambda)]=(1,3,2,4)\nonumber\\
&&R_5=\left (0,D\right);\;\;\;{\cal L}[\theta(\lambda)]=(1,3,4,2)\nonumber\\
&&R_6=\left (D,1.99D\right);\;\;\;{\cal L}[\theta(\lambda)]=(3,1,4,2)\nonumber\\
&&R_7=\left (1.99D,\infty\right);\;\;\;{\cal L}[\theta(\lambda)]=(3,4,1,2)
\end{eqnarray}
There are six crossing points near $\lambda =0$ (at $-3D, -2D -0.33D, 0, D, 1.99D$), which indicate
that drastic physical changes occur in that region. The fact that
there are no crossing points far from $\lambda =0$, indicates that only mild physical changes occur there.
This conclusion is the same as the conclusion derived earlier using 
a different generalized basis, and also using eigenvalues and eigenvectors. 
Again, if we average over the random variable $D$, then we get two comonotonicity regions
$(-\infty, 0)$ and $(0, \infty)$ as in the noiseless case.

\subsubsection{Shannon entropy in a generalized basis}

Let
\begin{eqnarray}
{\cal H} (\lambda)=\frac{1}{h_1(\lambda)+h_2(\lambda)}\begin{pmatrix}
h_1(\lambda) &h_3(\lambda)\\
[h_3(\lambda)]^*&h_2(\lambda)\\
\end{pmatrix}
\end{eqnarray}
be a positive semidefinite Hamiltonian, where the $h_1(\lambda), h_2(\lambda)$ are real functions of the coupling parameter $\lambda$, and 
$h_3(\lambda)$ is a complex function of $\lambda$. 
We consider the pseudo-probabilities
\begin{eqnarray}
s _{\cal H} (i|\lambda)=\frac{d}{n}{\rm Tr}[{\cal H}(\lambda)\sigma(i)]
;\;\;\;\sum _{i=1}^ns _{\cal H} (i|\lambda)=1.
\end{eqnarray}
where $\{\sigma (i)\}$ is a generalized basis, 
and the corresponding entropy
\begin{eqnarray}\label{785}
E_n(\lambda)=-\sum _{i=1}^ns_{\cal H}  (i|\lambda)\log [s_{\cal H}  (i|\lambda)].
\end{eqnarray}
We also consider the von Neumann entropy
\begin{eqnarray}\label{568}
E_{vN}(\lambda)&=&-{\rm Tr}[{\cal H} (\lambda)\log {\cal H}(\lambda)].
\end{eqnarray}

\begin{proposition}
A necessary and sufficient condition for the eigenvalues of ${\cal H} (\lambda)$ to be equal to each other (and equal to $1/2$), is that
$h_1(\lambda)=h_2(\lambda)$ and $h_3(\lambda)=0$. If there exists a value $\lambda=\lambda _0$ which satisfies these conditions, then
 the entropies for this Hamiltonian are:
\begin{eqnarray}
 E_n(\lambda _0)=\log n;\;\;\;E_{vN}(\lambda _0)=\log 2.
\end{eqnarray}
\end{proposition}
\begin{proof}
The characteristic equation of the matrix ${\cal H}(\lambda)$ is
\begin{eqnarray}
(h_1-\mu)(h_2-\mu)-|h_3|^2=0.
\end{eqnarray}
The discriminant of this equation is
\begin{eqnarray}
\Delta=(h_1-h_2)^2+|h_3|^2.
\end{eqnarray}
The eigenvalues are equal to each other when the discriminant is equal to zero and this gives the conditions $h_1(\lambda)=h_2(\lambda)$ and $h_3(\lambda)=0$. If there exists a value $\lambda=\lambda _0$ which satisfies these conditions, the Hamiltonian at this value is
${\cal H} (\lambda _0)=\frac{1}{2}{\bf 1}$, and therefore 
\begin{eqnarray}
s _{\cal H} (i|\lambda _0)=\frac{2}{n}{\rm Tr}\left [\frac{1}{2}\sigma(i)\right ]=\frac{1}{n}.
\end{eqnarray}
From this follows that  $E_n(\lambda _0)=\log n$.  Also when the eigenvalues are equal to each other (and equal to $1/2$),
then $E_{vN}(\lambda _0)=\log 2$.
\end{proof}

We have explained earlier that when the two eigenvalues are equal to each other, the ground state changes abruptly from one state to another. In the proposition above we have shown that at this point the entropies $E_n(\lambda _0)$ (and also the $E_{vN}(\lambda _0)$) take their maximum values.

We normalize the Hamiltonian $\theta (\lambda)$ and 
also the `noisy Hamiltonian\rq{} $\phi (\lambda)$ in Eqs.(\ref{hami}),(\ref{54}), so that their trace is one:
\begin{eqnarray}\label{540}
{\theta} _1(\lambda)=\frac{{\theta} (\lambda)}{{\rm Tr}[{\theta} (\lambda)]};\;\;\;
{\phi _1} (\lambda)=\frac{\phi (\lambda)}{{\rm Tr}[\phi (\lambda)]}
\end{eqnarray}
We calculated the pseudo-probabilities $s_{\theta _1}(i|\lambda)$, $s_{\phi _1}(i|\lambda)$ for values of $\lambda$ close to zero so that these operators are positive semidefinite. 
We then calculated the  entropy $E_n $ with the generalized basis in Eq.(\ref{sigma}), and also  with the generalized basis in Eq.(\ref{sigma2})
(we denote them $E_3, E_4$ for the noiseless normalized Hamiltonian ${\theta} _1(\lambda)$, and $E_3^{\rm noise}, E_4^{\rm noise}$ 
for the noisy normalized Hamiltonian ${\phi _1} (\lambda)$, correspondingly).

We also calculated the von Neumann entropy $E_{vN}(\lambda)$ and $E_{vN}^{\rm noise}(\lambda)$, for ${\theta} _1(\lambda)$ and
${\phi _1} (\lambda)$, correspondingly.
There is an exact symmetry $E_{vN}(-\lambda)=E_{vN}(\lambda)$ for the von Neumann entropy.
For the entropy in Eq.(\ref{785}), there is an approximate symmetry $E _n(-\lambda)\approx E _n(\lambda)$, for small values of $\lambda$.

In table \ref{t2} we give the von Neumann entropy $E_{vN}/\log 2$, and the entropies $E_3/\log3$ and $E_4/\log4$, for various values of $\lambda$. 
We also give the quantities
\begin{eqnarray}\label{123}
\frac{E_{vN}-E_{vN}^{\rm noise}}{E_{vN}};\;\;\;\frac{E_3-E_3^{\rm noise}}{E_3};\;\;\;\frac{E_4-E_4^{\rm noise}}{E_4}.
\end{eqnarray}
It is seen that the entropies $E_n$ are more robust in the presence of noise, than the von Neumann entropy $E_{vN}(\lambda)$.
For the amounts of noise that we used, the von Neumann entropy has error of approximately $9\%$, and the other entropies have error less than $1\%$. 
We note that in the example that we considered, all quantities in Eq.(\ref{123}) take positive values.
This is because noise makes the eigenvalues more unequal (see Eq.(\ref{321})) and this decreases the entropy.

We conclude that the entropies associated with our generalized bases, are more robust in the presence of noise than 
the entropies associated with orthonormal bases.

\section{Discussion}

We introduced redundancy into the concept of basis in a $d$-dimensional Hilbert space.
We started with a total set of $n>d$ vectors, and renormalized it into a 
a generalized basis, which consists of $n$ density matrices that resolve the identity.
The renormalization formalism uses M\"obius operators, and is inspired by the Shapley methodology in cooperative game theory, 
as discussed in \cite{V2017} for the special case 
of $n=d^2$ coherent states. In the present paper we use an arbitrary $n$ in the region $d<n<d^2$.
The non-independence and redundancy in a generalized basis, is quantified with a Shannon type of entropy which takes values in the 
interval $(\log n-\log d, \log n)$.

We have shown that the merit of calculations in a generalized basis, is that the results are sensitive to physical changes and robust in the presence of noise. These two requirements may appear to be contradictive, but they are not, because noise affects the whole basis in an 
almost equal way, while physical changes affect some parts of the basis more than others.   
We have shown with examples, that
addition of noise in the coefficients of a vector in a generalized basis, does not change the vector significantly.

We have also applied the formalism to the study of the ground state of a system with the Hamiltonian in Eq.(\ref{hami}), which
is frequently used as an approximation to an infinite-dimensional system, operating in the subspace of the lowest two states.
The concepts `location index with respect to a generalized basis',  
and `comonotonicity intervals of the coupling parameter', have been used to detect 
drastic changes in the ground state of the system, as the coupling parameter changes.
It has been shown that the method is robust in the presence of noise.

The work extends the area of coherent states, POVMs and frames and wavelets, in a new direction.
It starts from any total set of $n>d$ vectors, and leads to $n$ mixed states that resolve the identity.
The method has been used only with finite-dimensional Hilbert spaces. 
However cooperative game theory, is also applied to a continuum of players (e.g. \cite {AS}), and this could be used to extend our 
methodology to infinite-dimensional Hilbert spaces. 
In this case the sums contain an infinite number of terms, and the challenge is to ensure that they converge.
 
We note that the present paper is not related to work on quantum game theory, which is game theory with the superposition principle.
Here we use the mathematical methodology of Shapley in cooperative game theory, to renormalize the vectors in a total set, into density matrices that resolve the identity.

\newpage
\begin{table}
\caption{The vector $\ket{V}$ in Eq.(\ref{44}), is represented with 3,4,2 component vectors, using the generalized bases in Eqs. (\ref{sigma}), (\ref{sigma2})
and the orthonormal basis in Eq.(\ref{E1}), correspondingly. Random numbers 
(uniformly distributed in the interval $[-0.5, 0.5]$) are added to these components as in Eq.(\ref{478}), and approximations to $\ket{V}$ are calculated. The corresponding errors $\epsilon _3$, $\epsilon _4$, $\epsilon _{\rm orth}$ are shown. 
Their diagonal parts ($\epsilon_{3D}$, $\epsilon_{4D}$) and non-diagonal parts ($\epsilon_{3ND}$, $\epsilon_{4ND}$) are also shown.
The calculation has been repeated five times, with different sets of random numbers.}
\def\arraystretch{3}
\begin{tabular}{|c|c|c|c|c|c|c|}\hline
$\epsilon _3$&$\epsilon_{3D}$&$\epsilon _{3ND}$&$\epsilon _4$&$\epsilon_{4D}$&$\epsilon _{4ND}$&$\epsilon _{\rm orth}$\\\hline
$0.212$&$0.058$&$-0.013$&$0.296$&$0.049$&$0.038$&$0.310$\\\hline
$0.245$&$0.075$&$-0.015$&$0.144$&$0.021$&$0$&$0.340$\\\hline
$0.181$&$0.032$&$0$&$0.088$&$0.025$&$-0.017$&$0.233$\\\hline
$0.187$&$0.026$&$0.008$&$0.204$&$0.018$&$0.022$&$0.383$\\\hline
$0.143$&$0.051$&$-0.030$&$0.066$&$0.019$&$-0.015$&$0.347$\\\hline
\end{tabular} \label{t1}
\end{table}

\begin{table}
\caption{Various entropies for  the Hamiltonians $\theta _1(\lambda)$ and
${\phi _1} (\lambda)$ in Eq.(\ref{540}) (the entropies in the latter case have the superfix `noise\rq{}).
$E_{vN}$ is the von Neumann entropy, $E_3$ is the entropy with respect to the generalized basis in Eq.(\ref{sigma}) , 
and $E_4$ is the entropy with respect to the generalized basis in Eq.(\ref{sigma2}).}
\def\arraystretch{2.5}
\setlength{\tabcolsep}{8pt}
\begin{tabular}{|c|c|c|c|c|c|c|}\hline
$\lambda$&$E_{vN}/\log2$&$E_3/\log3$&$E_4/\log4$&$\frac{E_{vN}-E_{vN}^{\rm noise}}{E_{vN}}$&$\frac{E_3-E_3^{\rm noise}}{E_3}$&
$\frac{E_4-E_4^{\rm noise}}{E_4}$\\\hline
$-0.4$&$0.754$&$0.977$&$0.987$&$0.098$&$0.019$&$0.007$\\\hline
$-0.3$&$0.866$&$0.987$&$0.992$&$0.094$&$0.019$&$0.007$\\\hline
$-0.2$&$0.941$&$0.994$&$0.996$&$0.092$&$0.018$&$0.008$\\\hline
$-0.1$&$0.985$&$0.998$&$0.999$&$0.091$&$0.016$&$0.009$\\\hline
$0$&$1$&$1$&$1$&$0.091$&$0.015$&$0.009$\\\hline
$0.1$&$0.985$&$0.998$&$0.999$&$0.091$&$0.011$&$0.009$\\\hline
$0.2$&$0.941$&$0.994$&$0.996$&$0.092$&$0.009$&$0.008$\\\hline
$0.3$&$0.866$&$0.987$&$0.992$&$0.094$&$0.005$&$0.008$\\\hline
$0.4$&$0.754$&$0.977$&$0.986$&$0.098$&$0$&$0.006$\\\hline
\end{tabular} \label{t2}
\end{table}

\end{document}